\documentclass{article}
\usepackage{amsmath,amssymb,amsthm,amsfonts}
\usepackage{mathtools}
\usepackage{bm}
\usepackage{bbm}
\usepackage{mathrsfs}
\usepackage{physics} 
\usepackage{graphicx}
\usepackage{xcolor}
\usepackage{tikz-cd}

\definecolor{oiBlue}{HTML}{0072B2}
\definecolor{oiVermilion}{HTML}{D55E00}
\definecolor{oiBlack}{HTML}{000000}

\newcommand{\id}{\mathbbm{1}}       
\newcommand{\hs}{\mathcal{H}}       
\newcommand{\ks}{\mathcal{K}}       
\newcommand{\Tc}{\mathcal{T}}       
\newcommand{\ii}{\mathsf{i}}        

\theoremstyle{plain}
\newtheorem{theorem}{Theorem}
\newtheorem{lemma}{Lemma}

\theoremstyle{definition}
\newtheorem{definition}{Definition}
\theoremstyle{remark}
\newtheorem{remark}{Remark}

\title{Quantum mechanics over real numbers fully reproduces standard quantum theory}
\author{Alan C. Maioli $^{1}$, Evaldo M. F. Curado $^{1,2}$, Jean-Pierre Gazeau$^{3,4}$}

\begin{document}

\maketitle
\noindent
$^{1}$ \quad Centro Brasileiro de Pesquisas F\'{\i}sicas\\
$^{2}$ \quad National Institute of Science and Technology for Complex Systems\\
\ \ \ \quad Rua Xavier Sigaud 150,  Rio de Janeiro, Brazil\\
$^{3}$\quad Université Paris Cité, CNRS, Astroparticule et Cosmologie, 75013 Paris, France \\ 
$^{4}$\quad Faculty of Mathematics, University  of Bia\l ystok, 15-245 Bia\l ystok, Poland 
\begin{abstract}
   Standard quantum mechanics employs complex Hilbert spaces, but whether
complex numbers are fundamental or merely convenient has long been
debated. For decades, real-valued equivalents were considered mathematically possible but cumbersome. However, a highly cited 2021 result claimed that any quantum theory based on real numbers is experimentally falsifiable
via network Bell experiments. Yet, it remains an open question whether this falsification applies to all real-valued theories. Here we show that this conclusion rests on an incomplete real formulation, and we present a rigorous real-valued framework that perfectly reproduces all predictions of standard quantum mechanics. We demonstrate that the standard real tensor product ($\otimes_{\mathbb{R}}$) used in previous no-go theorems is algebraically incompatible with the rich structure of conventional quantum mechanics. We present a real framework based on K\"{a}hler space and prove that it is
exactly isomorphic to established quantum mechanics via an
explicit bijection $\gamma$. 
The isomorphism extends to composite
systems through a symplectic composition rule $\otimes^{\ks}$ that
replaces the Kronecker product. Consequently, our formulation achieves the maximal $\mathrm{CHSH}_{3}$ violation of $6\sqrt{2}$ using purely real variables, demonstrating that the no-go theorem is specific to a particular real representation of states 
and operators and to the composition rule $\otimes_\mathbb{R}$ built upon it, neither of which extends to the present K\"{a}hler framework. These results demonstrate that complex numbers are not fundamentally required by nature; rather, they encode a deeper real geometric structure that governs quantum interference and entanglement, settling this long debate.
\end{abstract}

\section{Introduction}
The appearance of the imaginary unit $\ii=\sqrt{-1}$ in the Schr\"{o}dinger
equation has been a source of conceptual unease since the inception of
quantum mechanics.  Dyson captured this vividly: \textit{``One of the
most profound jokes of nature is the square root of minus one that
Schr\"{o}dinger put into his equation\ldots\ Suddenly it became a wave
equation instead of a heat conduction equation''}~\cite{Dyson2009}.
Concretely, every measurable quantity in classical mechanics is real,
and complex numbers appeared there only as a computational aid.  Their
apparent necessity in quantum theory - to encode interference,
entanglement, and unitary evolution - has therefore demanded explanation.

Early attempts at a real formulation, pioneered by Stueckelberg in the
1960s~\cite{Stueckelberg1960}, established a standard recipe: 
replace a complex $d$-dimensional Hilbert space by a real $2d$-dimensional one via
the substitution
\begin{equation}\label{eq:doubling}
  1 \leftrightarrow
  \id_2 = \begin{pmatrix}1&0\\0&1\end{pmatrix},
  \qquad
  \ii \leftrightarrow
  \tau = \begin{pmatrix}0&-1\\1&0\end{pmatrix},
\end{equation}
where $\tau^2 = -\id_2$ mirrors the defining relation $\ii^2 = -1$,
and $\tau$ is concretely the matrix of counterclockwise rotation by $\pi/2$
in the plane.  
This doubling map reproduces single-system statistics faithfully.
For composite systems, however, one must form the tensor product of two
such doubled spaces.  Using the standard Kronecker product
$\otimes_{\mathbb{R}}$ on the doubled spaces on real numbers does not recover the
complex tensor product $\otimes_{\mathbb{C}}$: the dimension count alone
shows $$\dim(\mathbb{R}^{2m}\otimes_{\mathbb{R}}\mathbb{R}^{2n})=4mn,$$
whereas $$\dim(\mathbb{C}^m \otimes_{\mathbb{C}} \mathbb{C}^n)=mn$$ over
$\mathbb{C}$ (equivalently $2mn$ over $\mathbb{R}$).  The two theories
diverge precisely in multi-partite scenarios.

It is perhaps remarkable that, nearly a century after the mathematical 
foundations of quantum mechanics were laid, the correct notion of tensor 
product for composite systems continues to generate confusion. 
The tensor product $\otimes$
is arguably the most conceptually subtle operation in the quantum formalism: familiar in its matrix guise as the Kronecker product, yet deceptively simple notation concealing the passage from independent to genuinely entangled degrees of freedom: unlike direct sums or operator products, 
it encodes the correlation structure of compound systems in a way that 
has no classical counterpart. Its definition is not merely a notational 
choice but carries deep algebraic content - specifically, it must be 
compatible with every piece of structure that the constituent spaces carry. 
In a complex Hilbert space, that structure includes the field $\mathbb{C}$ 
itself: the tensor product $\otimes_{\mathbb{C}}$ is a tensor product 
over $\mathbb{C}$, enforcing $\mathbb{C}$-linearity 
across subsystems. When one passes to a real description via a 
doubling map, this $\mathbb{C}$-linearity does not disappear; 
it re-emerges as compatibility with the complex structure $J$. 
Replacing $\otimes_{\mathbb{C}}$ directly by the Kronecker product 
$\otimes_{\mathbb{R}}$ on doubled spaces discards precisely this 
constraint, producing a strictly larger - and physically spurious - space. 
The algebraic incompatibility of $\otimes_{\mathbb{R}}$ with the complex structure $J$ of the doubled space becomes evident once made explicit, yet its consequences for no-go arguments were not previously noted~\cite{Renou2021}.

Renou et al.~\cite{Renou2021} exploited this divergence to claim an
experimental falsification of real quantum theory, demonstrating that
the $\mathrm{CHSH}_3$ inequality - tailored for a tripartite
entanglement-swapping network - admits a maximal quantum violation of
$6\sqrt{2}$ that, they argued, is unreachable by any real quantum theory (in their sense).
A subsequent experiment confirmed the prediction~\cite{Chen2022}.

It is important to note that the no-go theorem derived by Renou et al.~\cite{Renou2021} is internally consistent within their axiomatic framework of a very specific ``real quantum physics'' (RQP). However, here we show that the incompatibility identified by Renou et al.\ is not a property of real numbers as such, but of the (inadequate) specific composition rule $\otimes_{\mathbb{R}}$ that, we argue, is incompatible with a faithful real representation of quantum mechanics. The core observation,
already implicit in the Stueckelberg programme, is that the matrices
$\id_2$ and $\tau$ in Eq.~\eqref{eq:doubling} are not just a
notational convenience - they constitute a complex structure
$J=\tau\otimes\id_N$ on the real doubled space.  This structure
promotes the doubled space to a K\"{a}hler space
$(\mathbb{R}^{2N},g,\omega,J)$ \cite{Volovich2025a,Volovich2025b}, and the compatible composition rule for
K\"{a}hler spaces - the symplectic tensor product $\otimes^{\ks}$ - is
distinct from $\otimes_{\mathbb{R}}$ (which justifies the subscript notation).

Our main result is the following isomorphism theorem (proved in
Appendix):

\begin{theorem}[Isomorphism]\label{thm:main}
  Let $\hs$ be a complex Hilbert space of dimension $N$.  Define the
  K\"{a}hler space $\ks = (\mathbb{R}^{2N},g,\omega,J)$ as in Eq.~\eqref{eq:kahler} (below).
  The map $\gamma^{-1}:\hs\to\ks$ defined by
  Eq.~\eqref{eq:gammainv} is a bijection, with inverse $\gamma$ given
  by Eq.~\eqref{eq:gamma}.  Moreover,
  $$
    \gamma^{-1}(A\otimes_{\mathbb{C}} B)
    = \gamma^{-1}(A)\otimes^{\ks}\gamma^{-1}(B),
  $$
  so $(\hs,\otimes_{\mathbb{C}})$ and $(\ks,\otimes^{\ks})$ are
  isomorphic as \textup{(}monoidal\textup{)} quantum theories.
\end{theorem}

This result - which we substantiate with explicit matrix calculations,
including the full $\mathrm{CHSH}_{3}$ computation - directly contradicts
the Renou et al.\ no-go theorem.  It is consistent with, and provides
the explicit constructive counterpart to the recent independent results
of Hoffreumon and Woods~\cite{Hoffreumon2025,Hoffreumon2026}, who showed
from an operational/postulational perspective that a real quantum theory
with representation locality is possible, and that the key assumption of
Renou et al.\ (product-state independence of sources) is experimentally
untestable~\cite{Hoffreumon2026}.  Related works appeared
simultaneously~\cite{Hita2025,Volovich2025a,Volovich2025b}.

It is only fair to acknowledge that complex structure is not an 
accidental feature of quantum mechanics: it is deeply woven into 
the formalism from the outset. The Schr\"{o}dinger equation is 
intrinsically complex; unitary evolution, the superposition principle, 
and interference phenomena all rely on the full algebraic richness of 
$\mathbb{C}$. What our analysis clarifies, however, is that this 
complexity is not fundamental in the ontological sense - it is 
not a primitive ingredient that must be postulated independently of 
the real structure of the theory. Rather, $\mathbb{C}$ encodes a 
real geometric datum: the complex structure $J$ of a K\"{a}hler manifold, 
satisfying $J^2 = -\id$ and compatible with both the metric $g$ and 
the symplectic form $\omega$. In this light, the debate between 
``real'' and ``complex'' quantum mechanics is somewhat misleading: 
the two descriptions are not rival theories but dual languages for 
the same geometry. The question is not whether complex numbers appear, 
but whether their appearance is irreducible or whether it 
reflects an underlying real structure that can be made explicit. 
Our isomorphism theorem answers this question unambiguously in 
favour of the latter.

\section{K\"{a}hler space quantum mechanics}

\subsection{Definition of a K\"{a}hler space}

A K\"{a}hler space $\ks$ \cite{Volovich2025a,Volovich2025b} is a quadruplet $(\mathbb{V},g,\omega,J)$ where
$\mathbb{V}$ is a real vector space, $g$ is a positive-definite inner
product (metric), $\omega$ is a non-degenerate skew-symmetric bilinear
form (symplectic form), and $J:\mathbb{V}\to\mathbb{V}$ is a
complex structure satisfying $J^2=-\id$.  These are linked by the
fundamental compatibility relations
\begin{equation}\label{eq:kahler}
  g(x,y) = \omega(x,Jy),
  \qquad
  \omega(Jx,Jy) = \omega(x,y),
  \qquad x,y\in\mathbb{V}.
\end{equation}
Together, $(g,\omega,J)$ identifies  $\mathbb{V}$ with a complex Hilbert space $\hs$
over $\mathbb{C}$ through the bijection $ (a\id+bJ)v_\ks \leftrightarrow (a+\ii b) v_\hs $.
The K\"{a}hler space $\ks$ and the complex Hilbert space $\hs$ thus encode
identical physics in different but mutually translatable
languages.

\subsection{Realification and complexification maps}

For an $N$-dimensional complex Hilbert space $\hs$, the corresponding
K\"{a}hler space is $\ks = \mathbb{R}^{2N}$ with the complex structure
\begin{equation}\label{eq:J}
  J = \tau \otimes \id_N,
  \qquad
  \tau = \begin{pmatrix}0&-1\\1&0\end{pmatrix},
\end{equation}
where $\otimes\equiv \otimes_\mathbb{R}$.

Any linear operator $L=X+\ii Y$ on $\hs$ (with $X,Y$ real $N\times N$
matrices) has a K\"{a}hler space counterpart
\begin{equation}\label{eq:gammainv}
  \mathcal{L} = \gamma^{-1}(L) =
  \id_2\otimes\operatorname{Re}(L) + \tau\otimes\operatorname{Im}(L)
  = \begin{pmatrix}X&-Y\\Y&X\end{pmatrix}.
\end{equation}
The inverse map $\gamma:\ks\to\hs$ extracts the complex operator from a
K\"{a}hler block matrix:
\begin{equation}\label{eq:gamma}
  L = \gamma(\mathcal{L}) =
  \frac{1}{2}\!\left(
    \Tc[\mathcal{L}] +
    \Tc[(-\sigma_y\otimes\id_N)\,\mathcal{L}]
  \right),
\end{equation}
where $\Tc[A\otimes C]=(\Tr A)C$ is the tensor contraction (see
Appendix).  The same prescription applies to state
vectors: for $\ket{\psi}_{\hs}=\ket{R}+\ii\ket{I}$ with real component
vectors $\ket{R},\ket{I}\in \ \mbox{subspace}\ \sim \mathbb{R}^N$,
\begin{equation}\label{eq:statevec}
  \ket{\psi}_{\ks} = \gamma^{-1}(\ket{\psi}_{\hs}) =
  \begin{pmatrix}\ket{R}&-\ket{I}\\\ket{I}&\ket{R}\end{pmatrix},
\end{equation}
a $2N\times 2$ real matrix.  Bra-states map as
$\bra{\psi}_{\ks} = \id_2\otimes\bra{R} - \tau\otimes\bra{I}$.

\begin{remark}
The matrices $\id_2$ and $\tau$ in Eq.~\eqref{eq:doubling} are not a
second quantum system; they are the matrix representation of the complex
structure $J$.  Treating them as an additional physical qubit - as done
in some analyses~\cite{Feng2025} - conflates the mathematical encoding
with the physics.
\end{remark}

\subsection{Inner product and symplectic form}

In K\"{a}hler space the physical inner product is the Riemannian metric $g$,
not the formal matrix product $\braket{\psi_2}{\psi_1}_{\ks}$ (which
is not a real scalar).  Explicitly,
\begin{equation}\label{eq:metric}
  g\!\left(\ket{\psi_2}_{\ks},\ket{\psi_1}_{\ks}\right)
  = \tfrac{1}{2}\,\Tr\!\left[\braket{\psi_2}{\psi_1}_{\ks}\right]
  = \operatorname{Re}\!\braket{\psi_2}{\psi_1}_{\hs},
\end{equation}
while the symplectic form reads
\begin{equation}
  \omega\!\left(\ket{\psi_2}_{\ks},\ket{\psi_1}_{\ks}\right)
  = \tfrac{1}{2}\,\Tr\!\left[-J\braket{\psi_2}{\psi_1}_{\ks}\right]
  = \operatorname{Im}\!\braket{\psi_2}{\psi_1}_{\hs}.
\end{equation}
All expectation values and probabilities computed via $g$ in $\ks$ agree
precisely with those computed via $\langle\cdot|\cdot\rangle_{\hs}$ in
$\hs$ (see Section S3).

\section{Symplectic composition rule and isomorphism}

The failure of literal quantum mechanics on real numbers in multipartite settings
originates entirely in the composition rule for subsystems.  To see why,
note that the algebraic rule for complex multiplication is
$(X_A+\ii Y_A)\otimes(X_B+\ii Y_B)
 = (X_A\otimes X_B - Y_A\otimes Y_B)
 + \ii(X_A\otimes Y_B + Y_A\otimes X_B)$.
The standard Kronecker product on doubled real spaces does not
respect this rule; instead it mixes real and imaginary sectors
incoherently.

\subsection{The symplectic tensor product}

\begin{definition}[Symplectic tensor product]\label{def:symp}
  Let $\mathcal{L}_A = \gamma^{-1}(A)$ and $\mathcal{L}_B = \gamma^{-1}(B)$
  be K\"{a}hler representatives of operators
  $A = X_A+\ii Y_A$, $B = X_B+\ii Y_B$.
  Their symplectic composite is
  \begin{equation}\label{eq:symptensor}
    \mathcal{L}_{AB} = \mathcal{L}_A \otimes^{\ks} \mathcal{L}_B
    =
    \begin{pmatrix}
      X_A\otimes X_B - Y_A\otimes Y_B &
     -X_A\otimes Y_B - Y_A\otimes X_B \\
      X_A\otimes Y_B + Y_A\otimes X_B &
      X_A\otimes X_B - Y_A\otimes Y_B
    \end{pmatrix},
  \end{equation}
  equivalently
  $\mathcal{L}_{AB}
   = \id_2\otimes(X_A\otimes X_B - Y_A\otimes Y_B)
   + \tau\otimes(X_A\otimes Y_B + Y_A\otimes X_B)$.
\end{definition}

This rule encodes the complex multiplication law $(A\otimes C-B\otimes D)
+\ii(A\otimes D + B\otimes C)$ directly in the block-matrix structure.
The commutative diagram in Fig.~\ref{fig:diagram} captures the
equivalence: going round either path - first composing in $\ks$ or first
lifting to $\hs$ and then projecting back - yields the same result.

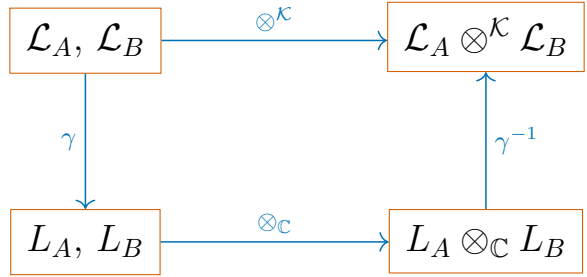
\begin{figure}[ht]
\centering
\Large
\begin{tikzcd}[
    row sep=1.8cm,
    column sep=3cm,
    cells={nodes={draw=oiVermilion, text=oiBlack, inner sep=6pt}},
    arrows={oiBlue}
]
  \mathcal{L}_A,\,\mathcal{L}_B
    \arrow[r,"\otimes^{\ks}"]
    \arrow[d,"\gamma"']
  & \mathcal{L}_A\otimes^{\ks}\mathcal{L}_B
  \\
  L_A,\,L_B
    \arrow[r,"\otimes_{\mathbb{C}}"]
  & L_A\otimes_{\mathbb{C}} L_B
    \arrow[u,"\gamma^{-1}"']
\end{tikzcd}
\caption{\textbf{Commutative diagram.}  The symplectic composition rule
  $\otimes^{\ks}$ is exactly equivalent to complexifying via $\gamma$,
  taking the standard complex tensor product $\otimes_{\mathbb{C}}$, and
  realifying via $\gamma^{-1}$.}
\label{fig:diagram}
\end{figure}

\subsection{Proof sketch of the isomorphism theorem}

The full proof occupies Section S4; we outline the key steps.

\emph{Bijection.}  Computing $\gamma(\gamma^{-1}(L))$ and
$\gamma^{-1}(\gamma(\mathcal{L}))$ using
Eqs.~\eqref{eq:gammainv}--\eqref{eq:gamma} directly yields $L$ and
$\mathcal{L}$ respectively, establishing the bijection.

\emph{Multiplicativity.}  Both $\gamma$ and $\gamma^{-1}$ satisfy
$\gamma( \mathcal{L}_A\,\mathcal{L}_B)=\gamma( \mathcal{L}_A)\gamma(\mathcal{L}_B)$, i.e.\ they
are ring homomorphisms (appendix).

\emph{Monoidal structure.}  The two key identities
\begin{equation}\label{eq:mono}
  \gamma(\mathcal{L}_A\otimes^{\ks}\mathcal{L}_B)
  = \gamma(\mathcal{L}_A)\otimes_{\mathbb{C}}\gamma(\mathcal{L}_B),
  \qquad
  \gamma^{-1}(A\otimes_{\mathbb{C}} B)
  = \gamma^{-1}(A)\otimes^{\ks}\gamma^{-1}(B),
\end{equation}
are established in appendix.  Together they
constitute the isomorphism of monoidal quantum theories. For the connection with the balanced tensor product (see section S8)

\section{Why Renou et al.'s construction fails}

The realification introduced in Eq.~(1) of Ref.~\cite{Renou2021}
associates to each complex density matrix $\rho$ a real matrix
$\rho^{\mathbb{R}}$ on a space of doubled dimension via the same
substitution $\ii\mapsto\tau$.  This is the map $\gamma^{-1}$
restricted to density matrices.  The crucial point, however, concerns
composite systems.  For a bipartite state $\rho_{AB}$, Renou et al.\
form the composite real space as $\hs_A^{\mathbb{R}}\otimes_{\mathbb{R}}\hs_B^{\mathbb{R}}$,
i.e.\ they use the standard Kronecker product on the doubled spaces.
This is not the symplectic product $\otimes^{\ks}$; it is an
incompatible operation that does not commute with the complex structure. It is important to stress that this divergence is not confined to the
choice of tensor product in the abstract sense: because $\otimes_{\mathbb{R}}$
and $\otimes^{\mathcal{K}}$ act on the same single-system blocks but combine
them differently, the resulting composite-system matrices and state vectors
are themselves different objects, living in real spaces of different
dimension ($4mn$ versus $2mn$), and obeying different algebraic relations.
The composition rule and the representation of the composite state are thus
two faces of the same failure, not independent choices.
Algebraically,
\begin{align}
  \rho^{\mathbb{R}}_{AB}
  &= \gamma^{-1}(\rho_A)\otimes_{\mathbb{R}}\gamma^{-1}(\rho_B)
  \;\neq\;
  \gamma^{-1}(\rho_A)\otimes^{\ks}\gamma^{-1}(\rho_B)
  = \gamma^{-1}(\rho_A\otimes_{\mathbb{C}}\rho_B).
\end{align}

The inequality is strict as soon as $\rho_A$ or $\rho_B$ has a non-zero
imaginary part.  As a result, the real composite space of Renou et al.\
is a strict enlargement of the physically relevant one:
$\dim(\hs_A^{\mathbb{R}}\otimes_{\mathbb{R}}\hs_B^{\mathbb{R}})
 = 4\dim(\hs_A)\dim(\hs_B)$,
versus
$\dim(\ks_A\otimes^{\ks}\ks_B)
 = 2\dim(\hs_A)\dim(\hs_B)$
(matching standard quantum mechanics).  The extra dimensions carry no
physical content; they arise from tensor-mixing the mathematical
encoding structure (the $\id_2/\tau$ block) with the physical degrees of
freedom.

\begin{remark}[Source independence]\label{rem:sources}
This structural point is closely related to the analysis of Hoffreumon--Woods~\cite{Hoffreumon2026}: the ``product-state independence'' of sources assumed in Ref.~\cite{Renou2021} requires source states to be Kronecker-product states in the doubled space. Our K\"{a}hler-compatible states $\gamma^{-1}(\ket{\psi_A})\otimes^{\ks} \gamma^{-1}(\ket{\psi_B})$ satisfy the physical (operational) independence condition but are not Kronecker-product states in the doubled space. This is why both approaches agree. Consequently, the assumption of real separability used by Renou et al.\ is strictly stronger than mere operational independence, in the sense that it is more restrictive. It imposes additional, hidden constraints that go beyond what is strictly necessary.
\end{remark}
\section{Bell-inequality violations}

\subsection{CHSH inequality}

Renou et al.\ showed~\cite{Renou2021} that the standard real formulation
can maximally violate the CHSH inequality, and this
remains true in our framework.  The K\"{a}hler space Pauli matrices are
(Appendix, Eqs.\ (S9))
\begin{equation}
  \sigma_x^{\ks} = \id_2\otimes\sigma_x,
  \quad
  \sigma_y^{\ks} = \tau\otimes\tau,
  \quad
  \sigma_z^{\ks} = \id_2\otimes\sigma_z.
\end{equation}
These satisfy the K\"{a}hler space anti-commutation relations
$\{\sigma_a^{\ks},\sigma_b^{\ks}\}=2\delta_{ab}\id_4$, inherited from
the complex algebra via $\gamma$.  A Bell state $\ket{\psi^-}_{\hs}$ maps to
$\ket{\psi^-}_{\ks}=\gamma^{-1}(\ket{\psi^-}_{\hs})$ (Eq.~\eqref{eq:statevec}),
and the standard CHSH combination yields
$g(\ket{\psi^-}_{\ks},\hat{C}_{\mathrm{CHSH}}^{\ks}\ket{\psi^-}_{\ks})
=2\sqrt{2}$, the Tsirelson bound.

\subsection{CHSH$_3$ inequality and the $6\sqrt{2}$ violation}

The $\mathrm{CHSH}_3$ functional for the tripartite entanglement-swapping
scenario (Alice, Bob, Charlie; two independent sources) was constructed
in Ref.~\cite{Renou2021} precisely to separate real and complex quantum
theories.  It is defined as
\begin{equation}
  \mathrm{CHSH}_3
  = \mathrm{CHSH}(1,2;1,2)+\mathrm{CHSH}(1,3;3,4)+\mathrm{CHSH}(2,3;5,6),
\end{equation}
with a sign adaptation functional $\mathscr{T}_b$ conditioned on Bob's
Bell-measurement outcome $b=(b_1b_2)$ (see Section S6) for the complete
definition).  The maximum value achievable by local-realist theories is
$6$; standard real quantum theory (with $\otimes_{\mathbb{R}}$) attains
at most $\approx 7.66$; complex quantum theory reaches $6\sqrt{2}\approx 8.49$.

In our K\"{a}hler space framework, for Bob's outcome $b=00$ and the Bell state
$\ket{\phi^+}$, the relevant operator in K\"{a}hler space is
\begin{equation}\label{eq:T00}
  \hat{\mathscr{T}}_{00}^{\ks}
  = 2\sqrt{2}\,\bigl(
    \sigma_z^{\ks}\otimes^{\ks}\sigma_z^{\ks}
    + \sigma_x^{\ks}\otimes^{\ks}\sigma_x^{\ks}
    - \sigma_y^{\ks}\otimes^{\ks}\sigma_y^{\ks}
  \bigr),
\end{equation}
with the Bell state mapped as
$\ket{\phi^+}_{\ks}=\gamma^{-1}(\ket{\phi^+})=\id_2\otimes\ket{\phi^+}$.
A direct computation (see Section S6) gives
\begin{equation}\label{eq:chsh3result}
  \bra{\phi^+}_{\ks}\,
  \hat{\mathscr{T}}_{00}^{\ks}\,
  \ket{\phi^+}_{\ks}
  = 6\sqrt{2}\,\id_2,
\end{equation}
so the metric expectation value is
$g(\ket{\phi^+}_{\ks},\hat{\mathscr{T}}_{00}^{\ks}\ket{\phi^+}_{\ks})
=\tfrac{1}{2}\Tc[6\sqrt{2}\,\id_2]=6\sqrt{2}$.
This is the maximum quantum value, achieved using only real
arithmetic.  The critical algebraic ingredient is the K\"{a}hler space identity
$\sigma_x^{\ks}\sigma_y^{\ks}=J\sigma_z^{\ks}$, where $J=\tau\otimes\id$
replaces the imaginary unit, together with the strict anti-commutation
$\{\sigma_x^{\ks},\sigma_y^{\ks}\}=0$ (see Section S6).  Both relations
are forbidden in the Renou et al.\ real formalism (which restricts to
real symmetric operators and $\otimes_{\mathbb{R}}$) but hold in $\ks$.

\section{Discussion}

Our results establish that the Renou et al.\ no-go theorem~\cite{Renou2021}
does not apply to K\"{a}hler space quantum mechanics.  The theorem is a
theorem about a particular real formulation - one that uses the
standard Kronecker product on doubled spaces - not about all possible
real formulations.  Our isomorphism theorem shows that a perfectly
adequate real formulation exists, provided one uses the algebraically
compatible composition rule $\otimes^{\ks}$.

\paragraph{Connection to Hoffreumon--Woods.}
Our results are strongly consonant with those of Hoffreumon and
Woods~\cite{Hoffreumon2025,Hoffreumon2026}.  Their 2025 paper shows
operationally that a real quantum theory with representation locality is
possible; our work provides the explicit constructive realisation (the
$\gamma$ map and $\otimes^{\ks}$ rule) and verifies the isomorphism by
direct computation.  Their 2026 paper additionally shows that the
``product-state independence'' assumption in Ref.~\cite{Renou2021} is
experimentally untestable---an independent argument for the same
conclusion.  Remark~\ref{rem:sources} explains how the two perspectives
are related.

\paragraph{The ``nonlocality'' objection.}
Feng, Ren, and Vedral~\cite{Feng2025} argued that any modified tensor product of the Hoffreumon--Woods type introduces a ``fundamental nonlocal map''. In our framework, the $\id_2/\tau$ block structure is the complex structure $J$ of the K\"{a}hler space---a fixed geometric attribute of the space, not a dynamical physical system. No physical ancilla is introduced; the apparent nonlocality is an artefact of interpreting the encoding structure as a physical degree of freedom. Local operations in $\hs$ (unitary evolution $U_A\otimes\id_B$, CPTP maps $\Phi_A\otimes\mathcal{I}_B$) map directly to local operations in $\ks$: $U_A^{\ks}\otimes^{\ks}\id_B^{\ks}$ acting on $\rho_{AB}^{\ks}$ (see Section S7). Crucially, this identical geometric logic applies to the mapping developed by Hoffreumon and Woods~\cite{Hoffreumon2025,Hoffreumon2026}; because their formulation also relies on a fixed encoding structure rather than a physical ancilla, the nonlocality critique is equally invalid for both results.

\paragraph{Geometric interpretation.}
Writing $\psi=q+\ii p$, the K\"{a}hler space is $\mathbb{R}^{2N}$ with the
triple $(g,\omega,J)$ satisfying $g(\cdot,\cdot)=\omega(\cdot,J\cdot)$.
The symmetry group is $U(N)=O(2N)\cap Sp(2N,\mathbb{R})$.  Complex
numbers in quantum mechanics are therefore not fundamental: they
encode the real geometric structure of a K\"{a}hler manifold governing
phase and composition.  This view harmonises with Volovich's independent
analysis~\cite{Volovich2025a,Volovich2025b}.

\section{Conclusions}
We have demonstrated:
(i) a bijection $\gamma$ between complex Hilbert space $\hs$ and K\"{a}hler
  space $\ks$;
(ii) an explicit symplectic composition rule $\otimes^{\ks}$ making the
  bijection monoidal;
(iii) the maximal $\mathrm{CHSH}_3$ violation $6\sqrt{2}$ in purely real
  arithmetic;
(iv) the precise reason why the Renou et al.\ construction fails (using
  $\otimes_{\mathbb{R}}$ in place of $\otimes^{\ks}$ is not a monoidal
  equivalence of theories).
Together these results establish that complex numbers are not
necessary for quantum mechanics: they encode a deeper real
geometric structure governing phase and composition.

To conclude, our approach should on no account be read 
as an attempt to expel complex numbers from quantum formalism.
The complex Hilbert space $\mathcal{H}$, or equivalently the real 
Kähler space $\mathcal{K}$, carries a fundamental involutive symmetry 
$\ii\mapsto -\ii$ (complex conjugation), whose geometric avatar in 
$\mathcal{K}$ is the involution $J\mapsto -J$ of the compatible complex 
structure. The physical ramifications of this single structural feature 
are far-reaching: Wigner's antiunitary time-reversal operator, charge 
conjugation, the CPT theorem, and the canonical decomposition of any 
operator into its self-adjoint and skew-adjoint parts all flow from it.
A bare real Hilbert space, lacking $J$ altogether, is entirely devoid 
of this structure and therefore cannot support these foundational 
symmetries of nature. What our framework establishes, rather, is that 
the indispensable role of $\ii$ is already fully encoded in the 
symplectic geometry of $\mathcal{K}$: complex structure is not imposed 
from without but is a \emph{derived} feature of the real, 
geometrically natural formulation. The imaginary unit is not 
discarded - it is \emph{explained}.

\section{Acknowledgments}
We thank the financial support from the Brazilian scientific
agencies Fundação Carlos Chagas Filho
de Amparo à Pesquisa do Estado do Rio de Janeiro (FAPERJ), Coordenação de Aperfeiçoamento
de Pessoal de Nível Superior (CAPES) and Conselho
Nacional de Desenvolvimento Científico e Tecnológico
(CNPq).

\section{Author contributions}
A.C.M., E.M.F.C., and J.-P.G. conceptualized the framework. A.C.M. developed the mathematical proofs and carried out the CHSH$_3$ calculations. All authors discussed the results, interpreted the geometric framework, and contributed to writing the manuscript.

\appendix
\section{Overview}\label{sec:S0}

This Appendices presents the detailed mathematical
framework supporting the main text.  We prove the isomorphism between
\textit{complex quantum mechanics}, i.e., QM formulated on a complex Hilbert space $\hs$)
and our real formulation (formulated on a K\"{a}hler  space $\ks$), and we carry
out the explicit Bell-inequality calculations.

\medskip\noindent
\textbf{Structure.}
Section~\ref{sec:S0} revisits  the concept of tensor product
Section~\ref{sec:S1} introduces tensor contraction.
Section~\ref{sec:S2} develops the K\"{a}hler space framework: definition, maps
$\gamma$ and $\gamma^{-1}$, bijection proof, fundamental relations, and
multiplicativity.
Section~\ref{sec:S3} defines the symplectic composition rule $\otimes^{\ks}$
and proves the two isomorphism lemmas.
Section~\ref{sec:S4} treats the CHSH inequality.
Section~\ref{sec:S5} treats the $\mathrm{CHSH}_3$ inequality.
Section~\ref{sec:S6} discusses local maps. Section~\ref{sec:S7} is about the balanced tensor product.

\medskip\noindent
\textbf{Notation.}  Throughout: $\id$ denotes the identity operator
(subscript indicates dimension when needed); $\tau=\begin{pmatrix}0&-1\\1&0\end{pmatrix}
= -\ii\sigma_y$; $\sigma_x,\sigma_y,\sigma_z$ are the usual Pauli matrices;
$\Tc$ denotes tensor contraction (Section~\ref{sec:S1}); $\ii=\sqrt{-1}$;
superscript $\ks$ on an operator or state denotes its K\"{a}hler space
representative.

\medskip\noindent
\textbf{Tensor product.}
 In essence, the Cartesian product, traditionally employed to model two systems independently in classical physics, undergoes a transformative process known as \textit{quantization}. This metamorphosis results in the emergence of the tensor product of two vector spaces, a fundamental framework indispensable for understanding the intricacies of quantum mechanical systems. 

\begin{center}
\emph{In what sense  tensor product is distinct of Cartesian product?}
\end{center}

  A vector space $V$ over a field $\mathbb{F}$, \textit{e.g.}, $\mathbb{Q}$ or $\mathbb{R}$ or $\mathbb{C}$,   is a set whose elements or vectors, may be added together and multiplied (``scaled'') by elements (``scalars'') in $\mathbb{F}$. Two essential properties must be satisfied: the distributivity of scalar multiplication with respect to the vector addition and the distributivity of scalar multiplication with respect to field addition. 
Then, the tensor product $V \otimes_\mathbb{F} W$ of two vector spaces $V$ and $W$ (over the \textbf{same} field) is the vector space over $\mathbb{F}$  consisting of all bilinear forms from $V \times W$ to $\mathbb{F}$. 

As a consequence,  $$\mathrm{dim}(V\otimes_\mathbb{F} W)= \mathrm{dim}V\mathrm{dim}W$$ while $$\mathrm{dim}(V\times W)= \mathrm{dim}V+\mathrm{dim}W.$$  

Usually,  one  simplifies $\otimes= \otimes_\mathbb{F}$ when there is no risk of confusion. On the other hand we will introduce the symbol $\otimes^{\ks}$ within the context of K\"{a}hler spaces defined in this material.  

The tensor product extends naturally to linear operators. 
If $A: V\to V$ and $B: W\to W$ are linear maps, their tensor product
$A\otimes_\mathbb{F} B: V\otimes_\mathbb{F} W \to V\otimes_\mathbb{F} W$
is the unique linear map defined on simple tensors by
\begin{equation}
  (A\otimes_\mathbb{F} B)(v\otimes w) = (Av)\otimes(Bw),
\end{equation}
and extended by linearity.  In matrix terms, if $A\in M_m(\mathbb{F})$
and $B\in M_n(\mathbb{F})$, then $A\otimes_\mathbb{F} B$ is the
$mn\times mn$ Kronecker product.  Crucially, the field $\mathbb{F}$
must be the same for both factors: this is precisely the point at issue
in the Renou et al.\ experiment, where $\mathbb{R}$-linear and
$\mathbb{C}$-linear tensor products are conflated.  Our K\"{a}hler framework
resolves this by working exclusively over $\mathbb{R}$ while encoding
the complex structure in the automorphism $J$, so that
$\otimes^{\ks}$ (Definition~\ref{def:symptp}) plays the role of
$\otimes_\mathbb{C}$ without ever leaving $\mathbb{R}$.
\section{Tensor contraction}\label{sec:S1}

In traditional quantum mechanics the partial trace is the canonical
reduction operation for density matrices.  We generalise it to
rectangular objects via the tensor contraction $\Tc$.

\begin{definition}[Tensor contraction]\label{def:tc}
  Let $M = A \otimes C$ where $A \in \mathbb{C}^{n\times n}$ is square and
  $C \in \mathbb{C}^{m\times k}$ is arbitrary.  The tensor contraction
  over the first factor is
  \begin{equation}\label{eq:sm:tc}
    \Tc(A \otimes C)
    = \sum_{i,j=1}^{n} \delta^{ij} (A)_{ij}\, C
    = (\Tr A)\, C.
  \end{equation}
  The operation is extended to sums of simple tensors by linearity.
\end{definition}

When $C$ is square, $\Tc$ reduces to the standard partial trace $\Tr_A$
over subsystem $A$.  When $C$ is rectangular, $\Tc$ provides the linear
extraction rule used in the complexification map $\gamma$.

\textbf{Key instances used below.}
$\Tc[\id_2\otimes C] = 2C$;
$\Tc[\tau\otimes C] = 0$.

\section{K\"{a}hler space framework}\label{sec:S2}

\subsection{Definition of a K\"{a}hler space}

\begin{definition}[K\"{a}hler space]\label{def:kahler}
A K\"{a}hler space \cite{Volovich2025a,Volovich2025b} is a quadruplet $(\mathbb{V},g,\omega,J)$ where:
\begin{itemize}[leftmargin=2em]
  \item $\mathbb{V}$ is a \textit{real} vector space;
  \item $g:\mathbb{V}\times\mathbb{V}\to\mathbb{R}$ is a
    positive-definite bilinear form (inner product/metric);
  \item $\omega:\mathbb{V}\times\mathbb{V}\to\mathbb{R}$ is a
    non-degenerate skew-symmetric bilinear form (symplectic form);
  \item $J:\mathbb{V}\to\mathbb{V}$ is a linear automorphism
    (complex structure) satisfying $J^2=-\id$.
\end{itemize}
These components obey the compatibility relations
\begin{align}
  g(x,y) &= \omega(x,Jy),\label{eq:sm:fund1}\\
  \omega(Jx,Jy) &= \omega(x,y),\label{eq:sm:fund2}
\end{align}
for all $x,y\in\mathbb{V}$.
\end{definition}


\subsection{Realification: from $\mathcal{H}$ to $\mathcal{K}$}

Let $L = X + \ii Y$ be a linear operator on $\hs\cong\mathbb{C}^N$,
where $X = \operatorname{Re}(L)$ and $Y = \operatorname{Im}(L)$ are
real $N\times N$ matrices.  Its K\"{a}hler space representative is
\begin{equation}\label{eq:sm:gammainv}
  \mathcal{L} = \gamma^{-1}(L)
  = \id_2\otimes X + \tau\otimes Y
  = \begin{pmatrix}X & -Y \\ Y & X\end{pmatrix}.
\end{equation}
Equivalently, using $L^* = X - \ii Y$:
\begin{equation}
  \gamma^{-1}(L)
  = \frac{1}{2}\!\left[
    \id_2\otimes(L+L^*) - \sigma_y\otimes(L-L^*)
  \right].
\end{equation}
The same formula applies to state vectors.  For
$\ket{\psi}_{\hs} = \ket{R}+\ii\ket{I}$ with
$\ket{R},\ket{I}\in\mathbb{R}^N$:
\begin{equation}\label{eq:sm:psikahler}
  \ket{\psi}_{\ks}
  = \gamma^{-1}(\ket{\psi}_{\hs})
  = \id_2\otimes\ket{R} + \tau\otimes\ket{I}
  = \begin{pmatrix}\ket{R} & -\ket{I} \\ \ket{I} & \ket{R}\end{pmatrix},
\end{equation}
a $2N\times 2$ real matrix ($\ket{\psi}_{\hs}$ is an $N\times 1$ column
vector, so $\ket{\psi}_{\ks}$ has dimensions $2N\times 2$).  The
conjugate bra-vector is
\begin{equation}
  \bra{\psi}_{\ks}
  = \bigl(\ket{\psi}_{\ks}\bigr)^\dagger
  = \id_2\otimes\bra{R} - \tau\otimes\bra{I}
  = \begin{pmatrix}\bra{R} & \bra{I} \\ -\bra{I} & \bra{R}\end{pmatrix}.
\end{equation}

\paragraph{Worked example: K\"{a}hler–Pauli matrices.}
The Pauli matrices $\sigma_x,\sigma_y,\sigma_z$ map as follows.
Since $\sigma_x$ and $\sigma_z$ are purely real,
\begin{equation}\label{eq:sm:KahPaux}
  \sigma_x^{\ks}
  = \id_2\otimes\sigma_x
  = \begin{pmatrix}0&1&0&0\\1&0&0&0\\0&0&0&1\\0&0&1&0\end{pmatrix},
  \quad
  \sigma_z^{\ks}
  = \id_2\otimes\sigma_z
  = \begin{pmatrix}1&0&0&0\\0&-1&0&0\\0&0&1&0\\0&0&0&-1\end{pmatrix}.
\end{equation}
For $\sigma_y = \ii\tau$ (so $\operatorname{Re}(\sigma_y)=0$,
$\operatorname{Im}(\sigma_y)=\tau$):
\begin{equation}\label{eq:sm:KahPauy}
  \sigma_y^{\ks}
  = \tau\otimes\tau
  = \begin{pmatrix}0&-1\\1&0\end{pmatrix}
    \otimes
    \begin{pmatrix}0&-1\\1&0\end{pmatrix}
  = \begin{pmatrix}0&0&0&1\\0&0&-1&0\\0&-1&0&0\\1&0&0&0\end{pmatrix}.
  \tag{\ref{eq:sm:KahPaux}$'$}
\end{equation}
\label{eq:sm:KahPauz}

The Kähler complex structure $J$ (Eq.~(4) of the main text) acts as
\begin{equation}\label{eq:sm:Jact}
  \sigma_x^{\ks}\sigma_y^{\ks} = J\sigma_z^{\ks},
  \qquad
  J = \tau\otimes\id_2,
\end{equation}
so $J$ replaces the imaginary unit $\ii$ in the Kähler-Pauli algebra.

\paragraph{Worked example: Kähler space qubit.}
Consider a general qubit
$\ket{\psi}_{\hs}=\cos(\theta/2)\ket{0}_{\hs}
 + e^{\ii\phi}\sin(\theta/2)\ket{1}_{\hs}$.
Writing $e^{\ii\phi}=\cos\phi+\ii\sin\phi$ we identify
$\ket{R}=\cos(\theta/2)\ket{0}_{\hs}+\cos(\phi)\sin(\theta/2)\ket{1}_{\hs}$
and
$\ket{I}=\sin(\phi)\sin(\theta/2)\ket{1}_{\hs}$.
Then $\ket{\psi}_{\ks}$ is the $4\times 2$ real matrix
\[
  \ket{\psi}_{\ks}
  = \begin{pmatrix}
      \cos(\theta/2) & 0 \\
      \cos\phi\sin(\theta/2) & -\sin\phi\sin(\theta/2) \\
      0 & \cos(\theta/2) \\
      \sin\phi\sin(\theta/2) & \cos\phi\sin(\theta/2)
    \end{pmatrix}.
\]
The matrices $\id_2$ and $\tau$ carry the complex structure of the K\"{a}hler space; they are not a physical qubit.

\subsection{Complexification: from $\mathcal{K}$ to $\mathcal{H}$}

\begin{definition}[Complexification map $\gamma$]
  For $\mathcal{L}\in\ks$, define
  \begin{equation}\label{eq:sm:gamma}
    L = \gamma(\mathcal{L})
    = \frac{1}{2}\!\left(
      \Tc[\mathcal{L}]
      + \Tc[(-\sigma_y\otimes\id_N)\,\mathcal{L}]
    \right).
  \end{equation}
  When $\mathcal{L}$ is block-square this reduces to
  $L = \tfrac{1}{2}(\Tr_1[\mathcal{L}]
    + \Tr_1[(-\sigma_y\otimes\id_N)\mathcal{L}])$.
\end{definition}

\textbf{Worked example.}
Applying $\gamma$ to $\sigma_y^{\ks}=\tau\otimes\tau$:
\begin{align*}
  \gamma(\sigma_y^{\ks})
  &= \tfrac{1}{2}\Bigl\{
    \Tc[\tau\otimes\tau]
    + \Tc[(-\sigma_y\otimes\id)(\tau\otimes\tau)]
  \Bigr\} \\
  &= \tfrac{1}{2}\Bigl\{
    0
    + \Tc[(-\sigma_y\tau)\otimes\tau]
  \Bigr\}.
\end{align*}

Using  $\sigma_y = \ii\tau$, so
$-\sigma_y\otimes\id = -\ii\tau\otimes\id$.
Then
$(-\sigma_y\otimes\id)(\tau\otimes\tau)
 = (-\ii\tau\cdot\tau)\otimes(\id\cdot\tau)
 = (-\ii\cdot(-\id_2))\otimes\tau
 = \ii\id_2\otimes\tau$.
Thus
$\Tc[\ii\id_2\otimes\tau] = \ii\cdot(\Tr\id_2)\cdot\tau = 2\ii\tau$,
and
$\gamma(\sigma_y^{\ks})=\tfrac{1}{2}\cdot 2\ii\tau=\ii\tau=\sigma_y$. \checkmark

\subsection{Bijection theorem}\label{subsec:bijection}

\begin{theorem}[Bijection]\label{thm:bijection}
  The maps $\gamma$ (Eq.~\eqref{eq:sm:gamma}) and $\gamma^{-1}$
  (Eq.~\eqref{eq:sm:gammainv}) are mutually inverse bijections between
  $\hs$ and $\ks$.
\end{theorem}

\begin{proof}
We prove the two statements.

\medskip
\noindent\textbf{Statement 1:}
$\gamma^{-1}(\gamma(\mathcal{L}))=\mathcal{L}$
for all $\mathcal{L}\in\ks$.

Let $\mathcal{L}=\id_2\otimes X+\tau\otimes Y$.
By linearity of $\Tc$:
\begin{align*}
  \Tc[\mathcal{L}]
  &= \Tc[\id_2\otimes X]+\Tc[\tau\otimes Y]
   = 2X + 0 = 2X,\\
  \Tc[(-\sigma_y\otimes\id)\mathcal{L}]
  &= \Tc[(-\sigma_y\otimes\id)(\id_2\otimes X)]
   + \Tc[(-\sigma_y\otimes\id)(\tau\otimes Y)]\\
  &= \Tc[-\sigma_y\otimes X] + \Tc[\ii\id_2\otimes Y]
   = 0 + 2\ii Y = 2\ii Y,
\end{align*}
where we used $(-\sigma_y)(\tau)=(-\ii\tau)(\tau)=-\ii(-\id_2)=\ii\id_2$.
Hence $\gamma(\mathcal{L})=\tfrac{1}{2}(2X+2\ii Y)=X+\ii Y$.
Applying $\gamma^{-1}$:
$\gamma^{-1}(X+\ii Y)=\id_2\otimes X+\tau\otimes Y=\mathcal{L}$.

\medskip
\noindent\textbf{Statement 2:}
$\gamma(\gamma^{-1}(L))=L$ for all $L\in\hs$.

Let $L=X+\ii Y$.  Then $\gamma^{-1}(L)=\id_2\otimes X+\tau\otimes Y$,
and by the calculation in Statement 1, $\gamma(\id_2\otimes X+\tau\otimes Y)
= X+\ii Y = L$.
\end{proof}

\subsection{Fundamental K\"{a}hler relations}\label{subsec:fund}

Define two states
$\ket{\psi_j}_{\hs}=\ket{R_j}+\ii\ket{I_j}$, $j=1,2$.
The Hilbert-space inner product is
\begin{equation}
  \braket{\psi_2}{\psi_1}_{\hs}
  = \bigl(\braket{R_2}{R_1}+\braket{I_2}{I_1}\bigr)
  + \ii\bigl(\braket{R_2}{I_1}-\braket{I_2}{R_1}\bigr),
\end{equation}
while its K\"{a}hler space version is
$\braket{\psi_2}{\psi_1}_{\ks}
 = \id_2\otimes(\braket{R_2}{R_1}+\braket{I_2}{I_1})
 + \tau\otimes(\braket{R_2}{I_1}-\braket{I_2}{R_1})$.
The \emph{metric} (real-valued inner product) in $\ks$ is
\begin{equation}\label{eq:sm:metric}
  g\!\bigl(\ket{\psi_2}_{\ks},\ket{\psi_1}_{\ks}\bigr)
  = \tfrac{1}{2}\,\Tc\!\bigl[\braket{\psi_2}{\psi_1}_{\ks}\bigr]
  = \operatorname{Re}\braket{\psi_2}{\psi_1}_{\hs},
\end{equation}
and the \emph{symplectic form} is
\begin{equation}\label{eq:sm:omega}
  \omega\!\bigl(\ket{\psi_2}_{\ks},\ket{\psi_1}_{\ks}\bigr)
  = \tfrac{1}{2}\,\Tc\!\bigl[-J\braket{\psi_2}{\psi_1}_{\ks}\bigr]
  = \operatorname{Im}\braket{\psi_2}{\psi_1}_{\hs},
\end{equation}
with complex structure $J=\tau\otimes\id_N$.
One directly verifies Eqs.~\eqref{eq:sm:fund1}--\eqref{eq:sm:fund2}.

\subsection{Multiplicativity of $\gamma$}\label{subsec:mult}

\begin{lemma}[Multiplicativity]\label{lem:mult}
  Both $\gamma$ and $\gamma^{-1}$ are ring homomorphisms:
  $\gamma(\mathcal{L}_A\mathcal{L}_B)=\gamma(\mathcal{L}_A)\gamma(\mathcal{L}_B)$ and
  $\gamma^{-1}(L_AL_B)=\gamma^{-1}(L_A)\gamma^{-1}(L_B)$.
\end{lemma}

\begin{proof}
Let $\mathcal{L}_A=\id_2\otimes X_A+\tau\otimes Y_A$ and
$\mathcal{L}_B=\id_2\otimes X_B+\tau\otimes Y_B$.
Then
\begin{align}
  \mathcal{L}_A\mathcal{L}_B
  &= (\id_2\otimes X_A)(\id_2\otimes X_B)
   + (\id_2\otimes X_A)(\tau\otimes Y_B)
   + (\tau\otimes Y_A)(\id_2\otimes X_B)
   + (\tau\otimes Y_A)(\tau\otimes Y_B)\notag\\
  &= \id_2\otimes(X_AX_B) + \tau\otimes(X_AY_B)
   + \tau\otimes(Y_AX_B) + \tau^2\otimes(Y_AY_B)\notag\\
  &= \id_2\otimes(X_AX_B-Y_AY_B)+\tau\otimes(X_AY_B+Y_AX_B),
  \label{eq:AksB}
\end{align}
using $\tau^2=-\id_2$.  Thus
$\gamma(\mathcal{L}_A\mathcal{L}_B)=(X_AX_B-Y_AY_B)+\ii(X_AY_B+Y_AX_B)
 = (X_A+\ii Y_A)(X_B+\ii Y_B)=\gamma(\mathcal{L}_A)\gamma(\mathcal{L}_B)$.

For $\gamma^{-1}$: let $L_A=X_A+\ii Y_A$, $L_B=X_B+\ii Y_B$.
Then $L_AL_B=X_AX_B-Y_AY_B+\ii(X_AY_B+Y_AX_B)$, and
$\gamma^{-1}(L_AL_B)=\id_2\otimes(X_AX_B-Y_AY_B)+\tau\otimes(X_AY_B+Y_AX_B)$,
which equals $\gamma^{-1}(L_A)\gamma^{-1}(L_B)$ by Eq.~\eqref{eq:AksB}.
\end{proof}

It is worth mentioning that $\mathcal{L}_A$ and $\mathcal{L}_B$ are not restricted to the same K\"{a}hler space, the only requirement is that the operation $\mathcal{L}_A\mathcal{L}_B$ must be well-defined. For instance $\mathcal{L}_A$ can be a linear operator and $\mathcal{L}_B$ a physical state. Also, the requirement for $L_A$ and $L_B$ is analogous.

\section{Symplectic composition rule and isomorphism}\label{sec:S3}

\subsection{Definition of $\otimes^{\ks}$}

\begin{definition}[Symplectic tensor product]\label{def:symptp}
  Let $\mathcal{L}_A=\gamma^{-1}(A)$, $\mathcal{L}_B=\gamma^{-1}(B)$
  with $A=X_A+\ii Y_A$, $B=X_B+\ii Y_B$.  Define
  \begin{equation}\label{eq:sm:symptp}
    \mathcal{L}_{AB}
    = \mathcal{L}_A\otimes^{\ks}\mathcal{L}_B
    =\begin{pmatrix}
       X_A\otimes X_B-Y_A\otimes Y_B &
      -X_A\otimes Y_B-Y_A\otimes X_B \\
       X_A\otimes Y_B+Y_A\otimes X_B &
       X_A\otimes X_B-Y_A\otimes Y_B
     \end{pmatrix},
  \end{equation}
  equivalently
  \begin{equation}\label{eq:sm:symptp2}
    \mathcal{L}_{AB}
    = \id_2\otimes(X_A\otimes X_B-Y_A\otimes Y_B)
    + \tau\otimes(X_A\otimes Y_B+Y_A\otimes X_B).
  \end{equation}
\end{definition}

\noindent
\textbf{Mnemonic.}  The rule $\otimes^{\ks}$ mimics matrix multiplication
on the $2\times 2$ block structure:
\[
  \begin{pmatrix}X_A&-Y_A\\Y_A&X_A\end{pmatrix}
  \otimes^{\ks}
  \begin{pmatrix}X_B&-Y_B\\Y_B&X_B\end{pmatrix}
  = \text{``matrix product of blocks''},
\]
where each ``product'' of blocks uses the Kronecker product $\otimes$ of
the constituent real matrices.

\noindent
\textbf{Commutative diagram.}
The equivalence $\gamma^{-1}\circ\otimes_{\mathbb{C}}
= \otimes^{\ks}\circ(\gamma^{-1}\times\gamma^{-1})$
is depicted in Fig.~\ref{fig:diagram}.

\subsection{Isomorphism lemmas}

\begin{lemma}[Realification of tensor product]\label{lem:S2}
  $\gamma^{-1}(L_A\otimes_{\mathbb{C}} L_B)
  = \gamma^{-1}(L_A)\otimes^{\ks}\gamma^{-1}(L_B)$.
\end{lemma}

\begin{proof}
Let $L_A=X_A+\ii Y_A$, $L_B=X_B+\ii Y_B$.  Then
\begin{align}
  \gamma^{-1}(L_A\otimes_{\mathbb{C}} L_B)
  &= \gamma^{-1}\bigl(
       (X_A+\ii Y_A)\otimes(X_B+\ii Y_B)
     \bigr)\notag\\
  &= \gamma^{-1}\bigl(
       (X_A\otimes X_B-Y_A\otimes Y_B)
       +\ii(X_A\otimes Y_B+Y_A\otimes X_B)
     \bigr)\label{eq:sm:lemma2step}\\
  &= \id_2\otimes(X_A\otimes X_B-Y_A\otimes Y_B)
     +\tau\otimes(X_A\otimes Y_B+Y_A\otimes X_B),\notag
\end{align}
and
\begin{align*}
  \gamma^{-1}(L_A)\otimes^{\ks}\gamma^{-1}(L_B)
  &= (\id_2\otimes X_A+\tau\otimes Y_A)
     \otimes^{\ks}
     (\id_2\otimes X_B+\tau\otimes Y_B)\\
  &= \id_2\otimes(X_A\otimes X_B-Y_A\otimes Y_B)
     +\tau\otimes(X_A\otimes Y_B+Y_A\otimes X_B),
\end{align*}
where the last equality uses Definition~\ref{def:symptp}.  The two
expressions are equal. \qed
\end{proof}

\begin{lemma}[Complexification of tensor product]\label{lem:S3}
  $\gamma(\mathcal{L}_A\otimes^{\ks}\mathcal{L}_B)
  = \gamma(\mathcal{L}_A)\otimes_{\mathbb{C}}\gamma(\mathcal{L}_B)$.
\end{lemma}

\begin{proof}
Direct computation using Definition~\ref{def:symptp} and
Eq.~\eqref{eq:sm:gamma}:
\begin{align*}
  \gamma(\mathcal{L}_A\otimes^{\ks}\mathcal{L}_B)
  &= \gamma\bigl(
      \id_2\otimes(X_A\otimes X_B-Y_A\otimes Y_B)
      +\tau\otimes(X_A\otimes Y_B+Y_A\otimes X_B)
     \bigr)\\
  &= (X_A\otimes X_B-Y_A\otimes Y_B)
     +\ii(X_A\otimes Y_B+Y_A\otimes X_B)\\
  &= (X_A+\ii Y_A)\otimes(X_B+\ii Y_B)
   = \gamma(\mathcal{L}_A)\otimes_{\mathbb{C}}\gamma(\mathcal{L}_B). \qed
\end{align*}
\end{proof}

Together with Theorem~\ref{thm:bijection} (bijection) and
Lemma~\ref{lem:mult} (multiplicativity), Lemmas~\ref{lem:S2}
and~\ref{lem:S3} establish Theorem~1 of the main text: $(\hs,\otimes_{\mathbb{C}})$
and $(\ks,\otimes^{\ks})$ are isomorphic monoidal quantum theories.

\section{CHSH inequality}\label{sec:S4}

We verify the CHSH inequality in K\"{a}hler space.
Choose Alice's operators $A_0=\sigma_z$, $A_1=\sigma_x$ and Bob's
operators $B_0=-(\sigma_x+\sigma_z)/\sqrt{2}$,
$B_1=(-\sigma_x+\sigma_z)/\sqrt{2}$ (all real-coefficient).
Because the coefficients are real, the K\"{a}hler maps are block-diagonal:
\begin{equation}
  C_{a,b}^{\ks}
  = \gamma^{-1}(A_a\otimes_{\mathbb{C}} B_b)
  = \begin{pmatrix}A_a\otimes B_b & 0 \\ 0 & A_a\otimes B_b\end{pmatrix}
  = \id_2\otimes(A_a\otimes B_b),
\end{equation}
a real $8\times 8$ matrix ($a,b\in\{0,1\}$).

The Bell state $\ket{\psi^-}_{\hs}=(\ket{10}-\ket{01})/\sqrt{2}$
maps to
\begin{equation}
  \ket{\psi^-}_{\ks}
  = \gamma^{-1}(\ket{\psi^-}_{\hs})
  = \id_2\otimes\ket{\psi^-},
\end{equation}
a $8\times 2$ real matrix (since $\ket{\psi^-}_{\hs}$ is $4\times 1$ over
$\mathbb{R}$, and the overall Kähler state is $8\times 2$; we write
$\id_2\otimes\ket{\psi^-}$ in the $8\times 2$ sense).

The metric expectation values are
$\langle C_{0,0}^{\ks}\rangle
 = \langle C_{1,0}^{\ks}\rangle
 = \langle C_{1,1}^{\ks}\rangle
 = 1/\sqrt{2}$
and $\langle C_{0,1}^{\ks}\rangle = -1/\sqrt{2}$.
Hence the CHSH combination reads
\begin{equation}
  \langle C_{0,0}^{\ks}\rangle
  +\langle C_{1,0}^{\ks}\rangle
  -\langle C_{0,1}^{\ks}\rangle
  +\langle C_{1,1}^{\ks}\rangle
  = 2\sqrt{2},
\end{equation}
achieving the Tsirelson bound~\cite{Tsirelson1987}.

\section{CHSH$_3$ inequality}\label{sec:S5}

\subsection{Setup}

The $\mathrm{CHSH}_3$ inequality is designed for an entanglement-swapping
network with three parties (Alice, Bob, Charlie) and two independent
sources.  Source~1 distributes particles to Alice and Bob; Source~2 to
Bob and Charlie.  Alice uses three measurement settings $x=1,2,3$;
Charlie uses six settings $z=1,\ldots,6$; Bob performs a Bell-state
measurement with four outcomes $b=b_1b_2\in\{00,01,10,11\}$,
corresponding to
$\{\ket{\phi^-},\ket{\psi^-},\ket{\phi^+},\ket{\psi^+}\}$ respectively,
where $\ket{\phi^\pm}=(\ket{00}\pm\ket{11})/\sqrt{2}$ and
$\ket{\psi^\pm}=(\ket{10}\pm\ket{01})/\sqrt{2}$.

\subsection{Operator definitions}

For Bob's outcome $b=(b_1,b_2)$, the functional is
\begin{align}
  \mathscr{T}_b(P)
  &= (-1)^{b_2}(S_{11}^b+S_{12}^b) + (-1)^{b_1}(S_{21}^b-S_{22}^b)
   \notag\\
  &\quad + (-1)^{b_2}(S_{13}^b+S_{14}^b)
     - (-1)^{b_1+b_2}(S_{33}^b-S_{34}^b)
   \notag\\
  &\quad + (-1)^{b_1}(S_{25}^b+S_{26}^b)
     - (-1)^{b_1+b_2}(S_{35}^b-S_{36}^b),
\end{align}
where $S_{xz}^b=\sum_{a,c}ac\,P(a,b,c|x,z)$.
This defines an operator $\hat{\mathscr{T}}_b$ via
$\mathscr{T}_b(P)=\bra{\psi}\hat{\mathscr{T}}_b\ket{\psi}$.
With Charlie's diagonal observables
$D_{ij}^C=(\sigma_i+\sigma_j)/\sqrt{2}$ and
$E_{ij}^C=(\sigma_i-\sigma_j)/\sqrt{2}$ ($i,j\in\{x,y,z\}$):
\begin{align}\label{eq:sm:opchsh3}
  \hat{\mathscr{T}}_b
  &= (-1)^{b_2}Z^A(D_{zx}^C+E_{zx}^C)
   + (-1)^{b_1}X^A(D_{zx}^C-E_{zx}^C)
   \notag\\
  &\quad + (-1)^{b_2}Z^A(D_{zy}^C+E_{zy}^C)
     - (-1)^{b_1+b_2}Y^A(D_{zy}^C-E_{zy}^C)
   \notag\\
  &\quad + (-1)^{b_1}X^A(D_{xy}^C+E_{xy}^C)
     - (-1)^{b_1+b_2}Y^A(D_{xy}^C-E_{xy}^C),
\end{align}
and Alice's operators $Z^A, X^A, Y^A$ corresponds to Pauli matrices $\sigma_z, \ \sigma_x,\ \sigma_y$, respectively.
Bounds: local-realist $\leq 6$; real QT of Ref.~\cite{Renou2021}
$\leq 7.66$; complex QT $= 6\sqrt{2}\approx 8.49$.

\subsection{Key algebraic obstruction in the Renou formalism}

The Renou et al.\ proof uses self-testing of local Pauli operators
$X^C, Y^C, Z^C$ from measurement statistics, that are defined as $X^C=(D_{zx}^C-E_{zx}^C)/\sqrt{2}\ $, $Y^C=(D_{zy}^C-E_{zy}^C)/\sqrt{2}\ $, and $Z^C=(D_{zx}^C+E_{zx}^C)/\sqrt{2}$.  The crucial
anti-commutation $\{X^C,Y^C\}=0$ and the product rule
$X^CY^C=\ii Z^C$ cannot simultaneously hold for real symmetric matrices
(the only matrices allowed in their formalism with $\otimes_{\mathbb{R}}$).

In our framework these relations hold via the Kähler-space replacements:
\begin{equation}\label{eq:sm:ksalg}
  \sigma_x^{\ks}\sigma_y^{\ks}
  = (\id_2\otimes\sigma_x)(\tau\otimes\tau)
  = \tau\otimes(\sigma_x\tau)
  = J\cdot(\id_2\otimes\sigma_z^{\phantom{\ks}})
  = J\sigma_z^{\ks},
\end{equation}
and
\begin{equation}
  \{\sigma_x^{\ks},\sigma_y^{\ks}\}
  = J\sigma_z^{\ks}+(-J\sigma_z^{\ks})
  = 0,
\end{equation}
where $J=\tau\otimes\id$ is the complex structure.  These are real matrix
equations; no complex numbers appear.

\subsection{Explicit computation: $b=00$}

For $b=00$, Eq.~\eqref{eq:sm:opchsh3} reduces to
\begin{align}
  \hat{\mathscr{T}}_{00}
  &= \sqrt{2}\bigl(
    \sigma_z\otimes\sigma_z
    + \sigma_x\otimes\sigma_x
    + \sigma_z\otimes\sigma_z
    - \sigma_y\otimes\sigma_y
    + \sigma_x\otimes\sigma_x
    - \sigma_y\otimes\sigma_y
  \bigr)\notag\\
  &= 2\sqrt{2}\bigl(
    \sigma_z\otimes\sigma_z
    + \sigma_x\otimes\sigma_x
    - \sigma_y\otimes\sigma_y
  \bigr).\label{eq:sm:T00}
\end{align}
In complex quantum mechanics,
$\bra{\phi^+}\hat{\mathscr{T}}_{00}\ket{\phi^+}=6\sqrt{2}$.

In K\"{a}hler space, the operator maps as
\begin{align}
  \hat{\mathscr{T}}_{00}^{\ks}
  &= 2\sqrt{2}\bigl(
    \sigma_z^{\ks}\otimes^{\ks}\sigma_z^{\ks}
    + \sigma_x^{\ks}\otimes^{\ks}\sigma_x^{\ks}
    - \sigma_y^{\ks}\otimes^{\ks}\sigma_y^{\ks}
  \bigr),\label{eq:sm:T00ks}
\end{align}
with the Bell state mapped as
$\ket{\phi^+}_{\ks}
 = \gamma^{-1}(\ket{\phi^+})
 = \id_2\otimes\ket{\phi^+}$.

Since $\sigma_z$, $\sigma_x$, $\sigma_y$ all have real coefficients in
the Hilbert-space sense (or purely imaginary for $\sigma_y$), we can
apply Lemma~\ref{lem:S2} directly.  Each term transforms as
$\gamma^{-1}(\sigma_a\otimes\sigma_a)=\sigma_a^{\ks}\otimes^{\ks}\sigma_a^{\ks}$,
so
$\bra{\phi^+}_{\ks}\hat{\mathscr{T}}_{00}^{\ks}\ket{\phi^+}_{\ks}
 = \gamma^{-1}(\bra{\phi^+}\hat{\mathscr{T}}_{00}\ket{\phi^+})
 = \gamma^{-1}(6\sqrt{2})
 = 6\sqrt{2}\,\id_2$.

The metric expectation value (Eq.~\eqref{eq:sm:metric}) is therefore
\begin{equation}\label{eq:sm:chsh3result}
  g\!\bigl(\ket{\phi^+}_{\ks},\hat{\mathscr{T}}_{00}^{\ks}\ket{\phi^+}_{\ks}\bigr)
  = \tfrac{1}{2}\Tr[6\sqrt{2}\,\id_2]
  = 6\sqrt{2},
\end{equation}
the maximum quantum value, achieved with purely real arithmetic. \qed

Additionally, the $\mathrm{CHSH}_3$ violation can be explicitly visualized for both the standard complex and real K\"{a}hler framework using the provided Mathematica supplement.

\section{Local maps in K\"{a}hler space}\label{sec:S6}

In complex quantum mechanics, a local unitary acting on subsystem $A$ is
$\rho'_{AB}=(U_A\otimes\id_B)\rho_{AB}(U_A^\dagger\otimes\id_B)$.
The same operation in K\"{a}hler space is
\begin{equation}\label{eq:sm:localmap}
  {\rho'}^{\ks}_{AB}
  = (U_A^{\ks}\otimes^{\ks}\id_B^{\ks})\,\rho^{\ks}_{AB}\,
    ({U_A^{\ks}}^\dagger\otimes^{\ks}\id_B^{\ks}),
\end{equation}
where $U_A^{\ks}=\gamma^{-1}(U_A)$, $\id_B^{\ks}=\id_2\otimes\id_B$,
and $\rho^{\ks}_{AB}=\gamma^{-1}(\rho_{AB})$.  By Lemma~\ref{lem:S2},
$\gamma(\rho'^{\ks}_{AB})=\rho'_{AB}$, confirming that local unitary
evolutions are represented locally in $\ks$ with respect to $\otimes^{\ks}$.

More generally, for a CPTP map $\Phi_A$ with Kraus operators $\{M_i\}$
satisfying $\sum_i M_i^\dagger M_i=\id_A$, the K\"{a}hler space action is
\begin{equation}
  {\rho'}^{\ks}_{AB}
  = \sum_i (M_i^{\ks}\otimes^{\ks}\id_B^{\ks})\,\rho^{\ks}_{AB}\,
    ({M_i^{\ks}}^\dagger\otimes^{\ks}\id_B^{\ks}),
\end{equation}
and the completeness relation becomes
$\sum_i {M_i^{\ks}}^\dagger M_i^{\ks}=\id_2\otimes\id_A=\id_A^{\ks}$
in $\ks$.  This confirms that local maps in $\hs$ correspond exactly to
local maps in $\ks$: there is no nonlocal overhead, contrary to the
claim of Ref.~\cite{Feng2025}.  The $\id_2/\tau$ block is part of the
fixed geometric structure (the complex structure $J$) and not a
dynamical ancillary system.

\section{Connection to the balanced tensor product}\label{sec:S7}

From the perspective of algebraic quantum theory, the symplectic product
$\otimes^{\ks}$ is the balanced tensor product over the complex
structure~\cite{Volovich2025a}:
\begin{equation}
  \hs_A\otimes_{(J_A,J_B)}\hs_B
  \coloneqq
  (\hs_A\otimes_{\mathbb{R}}\hs_B)\big/
  \langle J_A x\otimes y - x\otimes J_B y\rangle.
\end{equation}
Quotienting by this relation enforces $(J_A x)\otimes y = x\otimes(J_B y)$,
i.e.\ $\ii$-linearity across subsystems.  In block-matrix language this
quotient is precisely Definition 1.


\begin{thebibliography}{99}

\bibitem{Renou2021}
Renou, M.-O.\ et al.\
Quantum theory based on real numbers can be experimentally falsified.
\textit{Nature} \textbf{600}, 625--629 (2021).

\bibitem{Hoffreumon2025}
Hoffreumon, T.\ \& Woods, M.\,P.\
Quantum theory does not need complex numbers.
\textit{arXiv preprint}, arXiv:2504.02808 (2025).

\bibitem{Dyson2009}
Dyson, F.\,J.\
Birds and frogs.
\textit{Notices Am.\ Math.\ Soc.}\ \textbf{56}, 212--223 (2009).

\bibitem{Stueckelberg1960}
Stueckelberg, E.\,C.\,G.\
Quantum theory in real Hilbert space.
\textit{Helv.\ Phys.\ Acta} \textbf{33}, 727 (1960).

\bibitem{Chen2022}
Chen, M.-C.\ et al.\
Ruling out real-valued standard formalism in quantum theory.
\textit{Phys.\ Rev.\ Lett.}\ \textbf{128}, 040403 (2022).

\bibitem{Volovich2025a}
Volovich, I.\
Real quantum mechanics in a K\"{a}hler space.
\textit{arXiv preprint}, arXiv:2504.16838 (2025).

\bibitem{Volovich2025b}
Aref'eva, I.\ \& Volovich, I.\
Notes on real quantum mechanics in a K\"{a}hler space.
\textit{arXiv preprint}, arXiv:2506.07632 (2025).

\bibitem{Hoffreumon2026}
Hoffreumon, T.\ \& Woods, M.\,P.\
Quantum theory based on real numbers cannot be experimentally falsified.
\textit{arXiv preprint}, arXiv:2603.19208 (2026).

\bibitem{Hita2025}
Barrios Hita, P.\ et al.\
Quantum mechanics based on real numbers: a consistent description.
\textit{Phys.\ Rev.\ Lett.}\ \textbf{136}, 240202 (2026).

\bibitem{Feng2025}
Feng, T., Ren, C.\ \& Vedral, V.\
Locality implies complex numbers in quantum mechanics.
\textit{arXiv preprint}, arXiv:2504.07808 (2025).

\bibitem{Lancaster2025}
Lancaster, J.\,L.\ \& Palladino, N.\,M.\
Testing the necessity of complex numbers in traditional quantum theory
with quantum computers.
\textit{Am.\ J.\ Phys.}\ \textbf{93}, 110--120 (2025).

\bibitem{Tsirelson1987}
Tsirel'son, B.\,S.\
Quantum analogues of the Bell inequalities.
\textit{J.\ Math.\ Sci.}\ \textbf{36}, 557--570 (1987).


\end{thebibliography}
\end{document}